\documentclass{birkmult}
 \newtheorem{thm}{Theorem}[section]
 
 \newtheorem{lem}[thm]{Lemma}
 
 \theoremstyle{definition}
 
 \theoremstyle{remark}

 \numberwithin{equation}{section}

\begin{document}
%
%
%
%
%
%
%
%
%
\title[Solutions to the BBM equation with viscosity]
 {Solitary waves, periodic and elliptic solutions to the Benjamin, Bona \& Mahony (BBM) \\equation modified by viscosity}
\author[Stefan C. Mancas, Harihar Khanal]{Stefan C. Mancas, Harihar Khanal}

\address{%
Department of Mathematics\\
Embry--Riddle Aeronautical University\\
Daytona Beach, FL 32114--3900}
\email{stefan.mancas@erau.edu, harihar.khanal@erau.edu}

\author{Shardad G. Sajjadi}
\address{Center for Geophysics and Planetary Physics\br
Embry--Riddle Aeronautical University\br
Daytona Beach, FL 32114--3900\br}
\email{shardad.sajjadi@erau.edu}


\date{January 14, 2010}

\begin{abstract}
In this paper, we use a traveling wave reduction or a so-called spatial approximation to comprehensively investigate periodic and solitary wave solutions of the modified Benjamin, Bona \& Mahony equation (BBM) to include both dissipative and dispersive effects of viscous boundary layers.
Under certain circumstances that depend on the traveling wave velocity, classes of periodic and solitary wave like solutions are obtained in terms of Jacobi elliptic functions. 
An ad-hoc theory based on the dissipative term is presented, in which we have found a set of solutions in terms of an implicit function. Using dynamical systems theory we prove that the solutions of \eqref{BBMv} experience a transcritical bifurcation for a certain velocity of the traveling wave. Finally, we present qualitative numerical results.
\end{abstract}

\maketitle
\section{Introduction}
The first published observation about a solitary wave propagating along a uniform canal was made by John Scott Russel in 1834, who observed a solitary wave with an amplitude of about a foot and a width of about thirty feet on the Edinburgh to Glasgow canal. In his report to the British Association \cite{Russel}, Russel describes how he followed the solitary wave for more than one mile on horseback, observing that the observed solitary wave displayed a remarkable property of permanence preserving its original shape. He also noticed that taller waves will travel faster, and that solitary waves that move with different speeds, undergo a nonlinear interaction from which they emerge in their original shape.

In 1872, Joseph Boussinesq proposed a variety of possible models for describing the propagation of water waves in shallow channels \cite{Bouss}, including what is referred to as the Korteweg de Vries (KdV) equation \eqref{KdV}.

In 1895, Korteweg and de Vries, under the assumption of small wave amplitude and large wavelength of inviscid and incompressible fluids, derived an equation for the water waves, now known as the KdV equation, which also serves as a justifiable model for long waves in a wide class of nonlinear dispersive systems. KdV it has been also used to account adequately for observable phenomena such as the interaction of solitary waves and dissipationless undular shocks.
For the water wave problem, \eqref{KdV} is nondimensionalized, since the physical parameters  $|u|=3 \eta /2H$, $x=\sqrt6 x^*/H$, and $t=\sqrt{6g/H}t^*$, where $\eta$ is the vertical displacement of the free surface, $H$ is the depth of the undisturbed water, $x^*$ is dimensional distance and $t^*$ the dimensional time are all scaled into the definition of nondimensional space $x$, time $t$, and water velocity $u(x,t)$. When the physical parameters and scaling factors are appropriately absorbed into the definitions of $u$, $x$ and $t$, the KdV equation is obtained in the tidy form
\begin{equation}\label{KdV}
u_t+u_x+uu_x+u_{xxx}=0.
\end{equation}
A further reduction could be made by removing the second term of \eqref{KdV} by taking $x'=x-t$ and $t$ as independent variables, but nothing of significance is accomplished by this. Eq. \eqref{KdV}, or its equivalent without the second order term, is commonly taken as the staring point for mathematical studies of long-wave phenomena, although facts with considerable theoretical significance are already entailed in the derivation of KdV, \cite{Ben71}. Thus, the condensed form tends to disguise the meaning of the theory of the equation with regard to the original physical problem.

Although \eqref{KdV} has some remarkable properties it manifests non-physical properties; the most noticeable being unbounded dispersion relation. It is helpful to recognize that the main difficulties presented by \eqref{KdV} arise from the dispersion term and arise in  the linearized form
\begin{equation}\label{lin}
u_t+u_x+u_{xxx}=0
\end{equation} 
First note that when the solution of \eqref{lin} is expressible as a summation of Fourier components in the form $F(k)e^{-i (k x+\omega t)}$, the dispersion relation is
\begin{equation}\label{disp}
\omega=k-k^3.
\end{equation}
The phase velocity $\omega/k$ becomes negative for $k^2>1$, in contradiction of the original assumption of the forward traveling waves. More significantly, the group velocity $\frac{d\omega}{d k}=1-3k^2$ has no lower bound.

To circumvent this feature, it has been shown by Benjamin, Bona \& Mahony \cite{Ben72} that  \eqref{KdV} has an alternative format, which is called the regularized long-wave or BBM equation
\begin{equation}\label{BBM}
u_t+u_x+uu_x-u_{xxt}=0,
\end{equation}
in which the dispersion term $u_{xxx}$ is replaced by $-u_{xxt}$ that results in a bounded dispersion relation. That was utilized to prove existence, uniqueness, and regularity results \cite{Ben72}. It is contended that \eqref{BBM} is in important respects the preferable model over \eqref{KdV} which is unsuitably posed model for long waves. They also showed that  \eqref{BBM} has the same formal justification and possesses similar properties as that of KdV and its solutions generally have better smoothness properties that those of \eqref{KdV}.
The linearized version of  \eqref{BBM}  has the  dispersion relation
\begin{equation}\label{disp2}
\omega=\frac{k}{1+k^2},
\end{equation}
according to which both the phase velocity $\omega/k$ and the group velocity $\frac{d\omega}{d k}$ are bounded for all $k$. Moreover, both velocities approach zero for large $k$, which implies that fine scale features of the solution tend not to propagate.
The preference of \eqref{BBM} over \eqref{KdV} became clear in \cite{Ben72}, when the authors attempted to formulate an existence theory for \eqref{KdV}, respective to nonperiodic initial condition $u(x,0)$ defined on $(-\infty,\infty)$. In contrast, the existence and stability theory for \eqref{BBM} is essentially straight forward and perhaps the most persuasive evidence that \eqref{BBM} is better founded than \eqref{KdV} as a convenient model.

A generalization to \eqref{BBM} to include a viscous term is provided by the equation
\begin{equation}\label{BBMv}
u_t+u_x+uu_x-u_{xxt}=\nu u_{xx}
\end{equation}
where $\nu$ is transformed  kinematic viscosity coefficient of a liquid.
Here, we will study the above equation and henceforth we shall refer to \eqref{BBMv} as the modified BBM.

The structure of the paper is as follows. In \S2 a traveling wave solution of  \eqref{BBMv} will be derived. Depending on the traveling wave velocity $c$, three types of analytical solutions are found. When the viscosity is present a solution in implicit form is found using an ad-hoc theory for which the results will be shown numerically. In \S3 we will closely look at the principles of linearized stability for the traveling wave solution using dynamical systems theory for the modified BBM.  We will conclude the paper with the discussions and conclusions in \S4.

\section{Traveling solutions to the modified BBM}
The class of solitons solutions is found by employing the form  of the traveling wave solutions of \eqref{BBMv} which takes the form of the Ansatz
\begin{equation}\label{4}
u(x,t)=\phi (\zeta)
\end{equation}
where $\zeta\equiv x-ct$
is the traveling wave variable, and $c$ is a non negative translational wave velocity. The case when  $c<0$, i.e., solitons traveling in the opposite direction, can be treated in a similar fashion by letting $c\rightarrow -c$. The substitution of \eqref{4} in \eqref{BBMv} leads, after some simplification to
\begin{equation}\label{5}
(1-c)\phi_\zeta+\frac12 (\phi^2)_\zeta+c\phi_{\zeta\zeta\zeta}-\nu \phi_{\zeta\zeta}=0
\end{equation}
Assuming that $\phi,\phi_{\zeta} \rightarrow 0$ as $|\zeta|\rightarrow \pm\infty$, and by integrating once \eqref{5} we obtain the Lienard equation
\begin{equation}\label{6}
2\phi_{\zeta\zeta}=\alpha \phi_\zeta-\beta \phi^2-\gamma \phi,
\end{equation}
where $\alpha=\frac{2 \nu}{c}$, $\beta =\frac 1 c$, and $\gamma=\frac{2(1-c)}{c}$, which can be put in the first order form by making the substitution  $\omega(\phi)=\phi_{\zeta}$
\begin{equation}\label{7}
2\omega \omega_\phi=\alpha \omega-\beta \phi^2-\gamma \phi.
\end{equation}
Then, by making the substitution $\eta(\phi)=\frac2 \alpha \omega(\phi)$, the last equation  may be written in the form
\begin{equation}\label{8}
\eta \eta_\phi-\eta=b \phi^2+a\phi,
\end{equation}
where $a=-\frac{2 \gamma}{\alpha^2}$, and $b=-\frac{2 \beta}{\alpha^2}$. We recognize this as the Abel's equation of the second kind written in the canonical form.

\subsection{No viscosity, $\nu=0$}
If the viscosity is not present, i.e., $\nu=0$, then $\alpha=0$, and hence eq. \eqref{5} becomes

\begin{align}
\phi_{\zeta\zeta}+\frac{1}{2c}\phi^2=&\frac{c-1}{c}\phi \qquad \mathrm{if} \qquad c>1 \\
\phi_{\zeta\zeta}+\frac{1+c}{2}\phi^2=&-c\phi \qquad \mathrm{if} \qquad 0<c<1 
\end{align}

with analytical  traveling waves solutions
\begin{align}\label{62}
u(x,t)=&3(c-1) \mathrm{Sech}^2\Big[{\frac{\sqrt{(c-1)/c}}{2}(x-c t)}\Big]\qquad \mathrm{if} \qquad c>1  \\
u(x,t)=&-\frac{3c}{1+c} \mathrm{Sec}^2\Big[{\frac{\sqrt {c}}{2}\Big (x-t/(1+c)\Big)}\Big] \qquad \mathrm{if} \qquad 0<c<1. 
\end{align}
In the $(x,t)$ space the first set of solutions will be solitary wave like and will move with a translational velocity $c>1$; we call these fast waves. The second set of solutions will be periodic, and  will have a translational velocity $c\in (0,1)$; we will call these the slow waves that are unbounded if $\sqrt c\Big (x-t/(1+c)\Big)=(2j+1)\pi$,
$j\in Z$. 

At the boundary between the fast and slow waves we will encounter periodic solutions that will be given in terms of elliptic functions that will travel with velocity $c=1$ as we will see next.

When $c=1$, \eqref{6} becomes
\begin{equation}\label{el1}
\frac{\phi^2}{2}+\phi_{\zeta\zeta}=0,
\end{equation}
which is not by all means a simpler equation.  By multiplying by $\phi_{\zeta}$ and integrating once again we obtain
\begin{equation}\label{el2}
\frac{\phi^3}{3}+(\phi_{\zeta})^2=\frac{A^3}{3},
\end{equation}
where A is some nonzero constant of integration. Now let's use the substitution $\mu^2=A-\phi$, where $\mu=\mu(\zeta)$. Hence \eqref{el2} becomes
\begin{equation}\label{el3}
\mu_{\zeta}= \pm\frac {\sqrt 3}{6} \sqrt{\mu^4-3A\mu^2+3A^2}.
\end{equation}
Moreover, let's assume $\mu(\zeta)=3^{1/4}A^{1/2}z(\zeta)$, then \eqref{el3} becomes
\begin{equation}\label{el4}
z_{\zeta}= \pm\frac {A^{1/2}}{3^{1/4}2} \sqrt{z^4-\sqrt{3}z^2+1}.
\end{equation}
This differential equation will be solved using Jacobian elliptic functions.\\
\begin{lem}
If $Z=1+2z^2 \cos {2\alpha}+z^4$, then $$u(x,k)=\int^{x}_{0}\frac{dz}{\sqrt{Z}}=\frac 12 sn^{-1}\frac{2x}{1+x^2},$$ with $k=\sin{\alpha}$ 
\end{lem}
\begin{proof}
Putting $z=\tan{\theta}$, we find $$u(x,k)=\int^{\tan^{-1}x}_{0}\frac{d \theta}{\sqrt{1-\sin^{2}\alpha \sin^{2}{2 \theta}}},$$ followed by  $y=\sin{2\theta}$, which in turn leads to
$$u(x,k)=\int^{\frac{2x}{1+x^2}}_{0}\frac{\frac 12 d y}{\sqrt{(1-y^2)(1-k^2 y^2)}},$$ where the integrand is the elliptic integral of the first kind \cite{Bowman}.
\end{proof}

Therefore, solving \eqref{el4}, we obtain the solution in $z$
\begin{equation}\label{el5}
\frac{2 z}{1+z^2}=\pm sn\big(A^{1/2}/3^{1/4}\zeta,k\big),
\end{equation}
where $k=\sin{\frac{5 \pi}{12}}=\frac{\sqrt{3}+1}{2 \sqrt{2}}$ is the modulus of the Jacobian elliptic function. 

It follows that in $\mu$ we will have
\begin{equation}\label{el6}
\frac{2 a \mu }{a^2+\mu^2}=\pm sn(3^{-1/2} a\zeta,k),
\end{equation}
where $a=3^{1/4}A^{1/2}$.
By solving the quadratic, we obtain
\begin{equation}\label{el7}
\mu=\pm a\frac{1\pm cn(3^{-1/2}a\zeta,k)}{sn(3^{-1/2}a\zeta,k)}.
\end{equation}

Hence, the analytical solution to the \eqref{BBM} without viscosity $\nu=0$ and traveling wave velocity $c=1$ is finally
\begin{eqnarray}\label{el8}
u(x,t)=&A\Big[1-3^{1/2}\Big(\frac{1\pm cn\big(A^{1/2}/3^{1/4}(x-t),k\big)}{sn\big(A^{1/2}/3^{1/4}(x-t),k\big)}\Big)^2\Big] \notag\\ 
=&A\Big[1-\sqrt{3
}\frac{1\mp cn\big(A^{1/2}/3^{1/4}(x-t),k\big)}{1 \pm cn\big(A^{1/2}/3^{1/4}(x-t),k\big)}\Big] .
\end{eqnarray}
These critical solutions are unbounded along the parallel lines in the $(x,t)$ plane where $x-t=\frac{2a}{nAK}$, $n\geq 0$, and $2K$ periodic with $K$ given by the complete elliptic integral
\begin{equation}
\frac K2=\int_0^{1}\frac{dz}{\sqrt{z^4-\sqrt{3}z^2+1}}=2.76806
\end{equation}

The expressions \eqref{62},(2.9) and \eqref{el8}, describe the whole class of solitary wave solutions, with spectrum $c \in(0,\infty)$; but we shall generally use the symbol $\phi$ to mean some specific member of the class, referring  to as the solitary wave or periodic solution, which is in concordance with \cite{Ben72}. They present periodic and solitary wave solutions to \eqref{BBM}, with no restrictions on the constant $A$, which may be advantageous when adapting  results to physical problems. The validity of the equation as an approximate model for water waves in real systems depends on the magnitude $|u|$ of solutions being every where small, and so it is warranted to require that the velocity $c$ should be chosen such that the waves  will travel at a speed that is asymptotically close to the linear  shallow water speed about the point $x=t$, or equivalently $x^*=\sqrt{gH}t^*$. Hence, $c$ can not be much larger than one.\\
The traveling waves solutions obtained by \eqref{62}, (fast waves) 
are plotted using
MATLAB function \verb+surf+ and shown in Fig.\ref{fig1}.
To plot the periodic  solution given by (2.9) \eqref{el8}
is more involved. We used MATLAB function \verb+ellipj+ to compute the
elliptic functions $cn$ and $sn$ of \eqref{el8}. These solutions have singularities along $\sqrt c\Big (x-t/(1+c)\Big)=(2j+1)\pi$ and $x-t=\frac{2a}{nAK}$ respectively, and therefore we plot them in different domains see Figs.\ref{fig2}, \ref{fig3}.
\begin{figure}[!ht]
					\begin{center}
\includegraphics[width=0.40\textwidth]{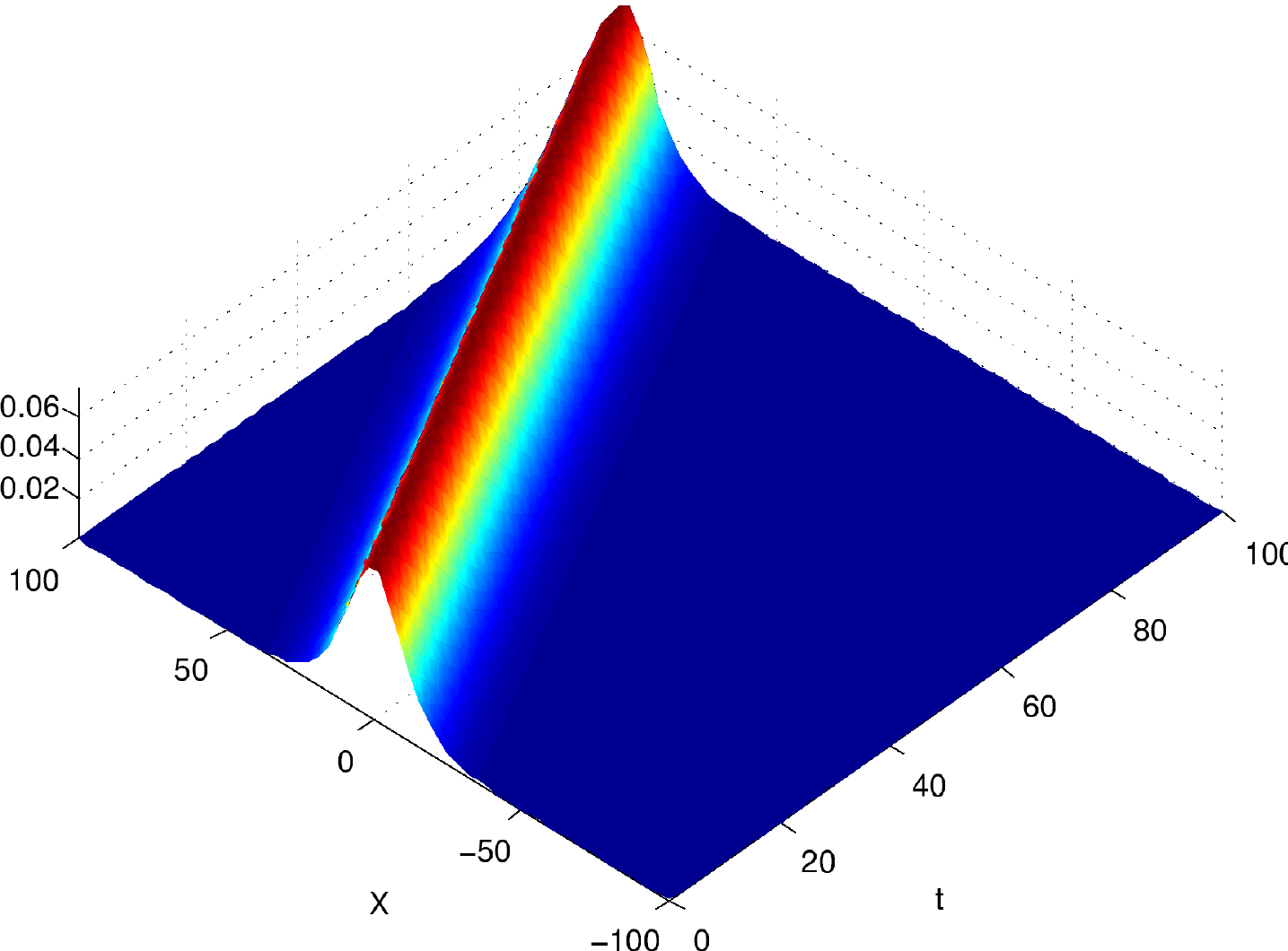}
					\end{center}
				\caption{Exact solutions of the \eqref{BBM} given by \eqref{62} with $c=1.025$.}\label{fig1}
\end{figure}
\begin{figure}[!ht]
					\begin{center}
\includegraphics[width=0.40\textwidth]{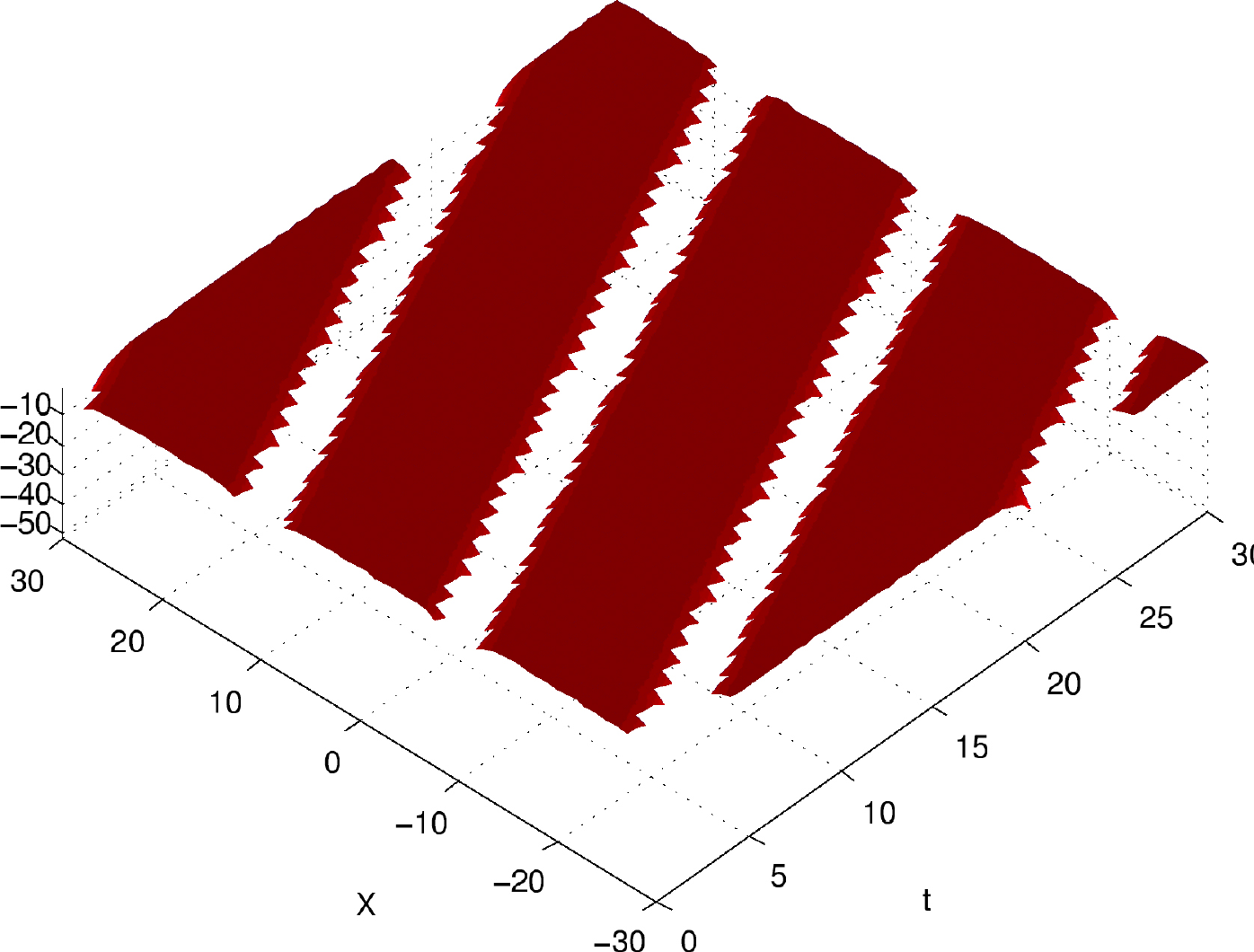}
					\end{center}
				\caption{Exact solutions of the \eqref{BBM} given by (2.9) with $c=0.1$.}\label{fig2}
\end{figure}
\begin{figure}[!ht]
					\begin{center}
\includegraphics[width=0.40\textwidth]{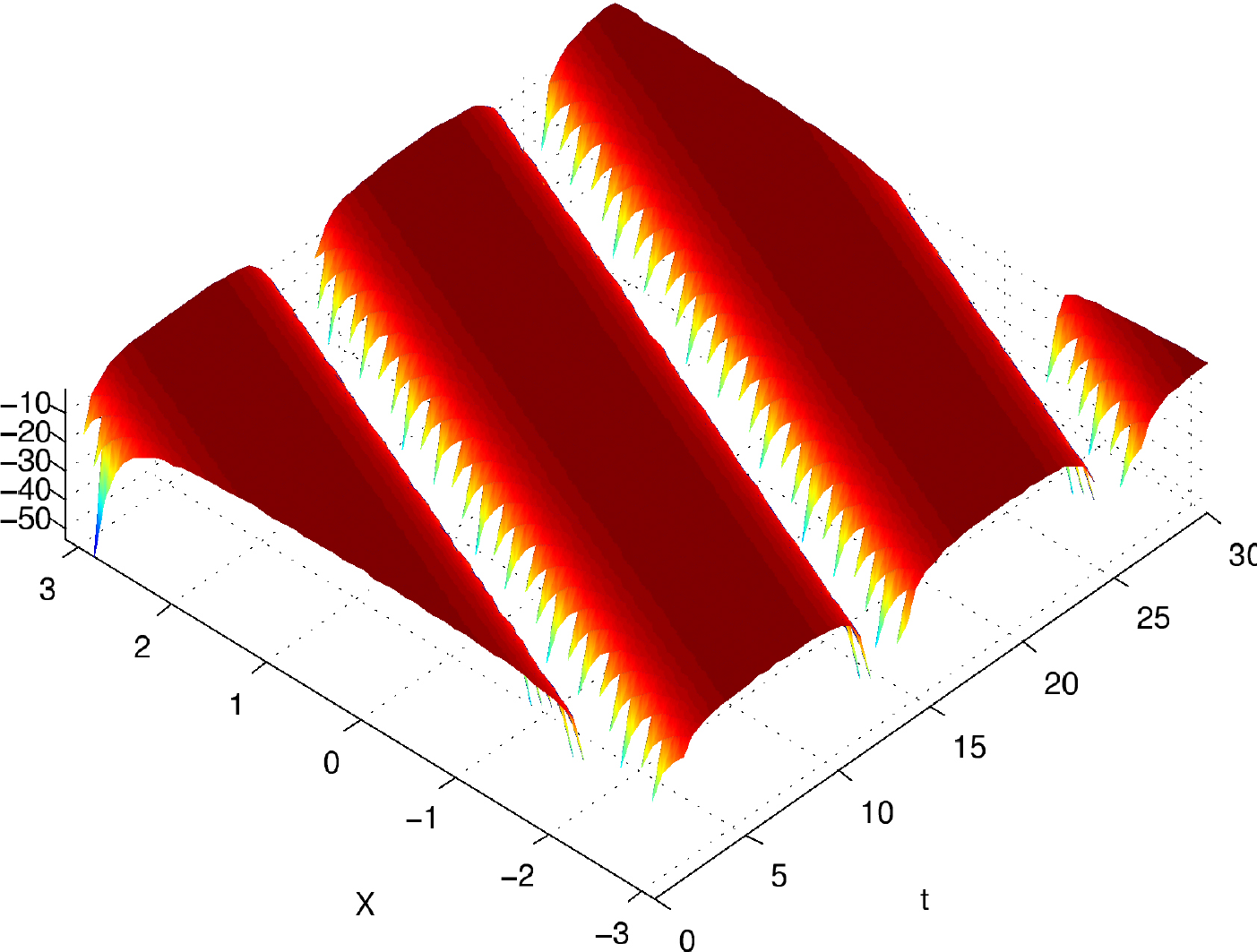}
					\end{center}
				\caption{Exact solutions of \eqref{BBM} given by \eqref{el8} with $c=1$.}\label{fig3}
\end{figure}

\subsection{Viscous layers present, $\nu>0$}

For the case when  viscous boundary layers are present, we will use 
an algebraic method to solve \eqref{BBMv}.
We will use at last two transformations on both the independent variable and  the dependent variable on the Abel's equation
\begin{align}
\tau&=\Big(1-\frac {\eta}{\phi}\Big)^2 \label{9}\\
\chi&=-2\Big(a+b \phi+\frac \eta \phi-\frac {\eta^2}{\phi^2}\Big)\label{10}
\end{align}
which will allow us to write \eqref{8} as
\begin{equation} \label{11}
\chi \chi_{\tau}-\chi=2 \big(\tau-2 \tau^{\frac 12}+1-a+a \tau^{-\frac 12}\big)
\end{equation}

According to \cite{Polyanin} the last equation has the closed form solution
\begin{equation}\label{12}
\chi(\tau)=\frac{a+(1-a)\tau^{\frac 12}-2 \tau +\tau^{\frac 32}}{3(\tau^{\frac 12}-1)(\tau-\tau^{\frac 12}-a)}\Big( 12 a \tau^{\frac 12} +12(1-a) \tau-8 \tau^{\frac 32}+3 \tau^2\Big)
\end{equation}
Now, the idea is to use \eqref{9}, \eqref{10}, and \eqref{12} to obtain an algebraic equation in one variable only. Two achieve this, let us equate \eqref{10} and \eqref{12} and use the same variable $x=\frac \eta \phi$ in all three mentioned equations. Therefore, we obtain
\begin{eqnarray}\label{13}
&\frac{a+(1-a)(1-x)-2(1-x)^2+(1-x)^3}{-3x[(1-x)^2-(1-x)-a]}= \notag\\ 
=&\frac{-2(a + b \phi +x(1-x))}{(1-x)[12a+12(1-a)(1-x)-8(1-x)^2+3 (1-x)^3]}
\end{eqnarray}
After few algebraic manipulations \eqref{13} can be written in the quartic form
\begin{equation}\label{14}
3 x^4-4 x^3-12 a x^2+6(2a-1)x+6 a+ 7 +6 b \phi=0.
\end{equation}
Now, we will use the fact that $x=\frac{2 \omega(\phi)}{\alpha \phi}$ in \eqref{14} and we obtain
\begin{equation}\label{15}
\omega^4(\phi)+a_3(\phi)\omega^3(\phi)+a_2(\phi)\omega^2(\phi)+a_1(\phi)\omega(\phi)+ a_0(\phi)=0,
\end{equation}
where the coefficients of the quartic are
\begin{align}\label{coef}
a_3(\phi)&=-\frac{4 \nu \phi}{3 c} \\
a_2(\phi)&=\frac{4 \phi^2(1-c)}{c}\notag \\
a_1(\phi)&=-\frac{2 \nu \phi^3\big(\nu^2+2c(1-c)\big)}{c^3}\notag \\
a_0(\phi)&=\frac{\nu^2 \phi^4\Big(7 \nu^2-3c\big(\phi+2(1-c)\big)\Big)}{3 c^4}. \notag
\end{align}
Since 
\begin{equation}\label{ode}
\omega(\phi)=d\phi/d\zeta,
\end{equation} 
we obtain the solution $\phi(x-ct)$ implicitly by solving the first order ODE (\ref{ode}) for 
$\phi(\zeta)$ with $\omega(\phi)$ as a root of the algebraic equation (\ref{15}).

First, we found the roots $\omega_{1234}(\phi)$ using MATLAB function
\verb+solve+. The expressions are highly nonlinear, extremely lengthy and
therefore not presented here. As such, the solutions of the ODE (\ref{ode})
for the nonzero values of $\nu$ take very long time to compute and due to the possible
subtractive cancellation the approximate solutions may not be satisfactory. 
However, we can compute the roots for the cases with $\nu=0$ and $\nu=1$ easily.
We solve the ODE (\ref{ode}) using the MATLAB function \verb+ode45+ using one of the 
quartic roots $\omega_{i}$ for some $i$ of equation (\ref{15}) for both viscosity parameters $\nu=0$ and $\nu=1$, 
and traveling velocity speed $c=1.025$ and $c=0.1$.
The numerical solutions $\phi(\zeta)$ are plotted as $\phi(x-ct)$ versus $t$, 
and they qualitatively agree well with the closed form solutions 
presented in \S 2.1.
\begin{figure}[!ht]
					\begin{center}
\includegraphics[width=0.40\textwidth]{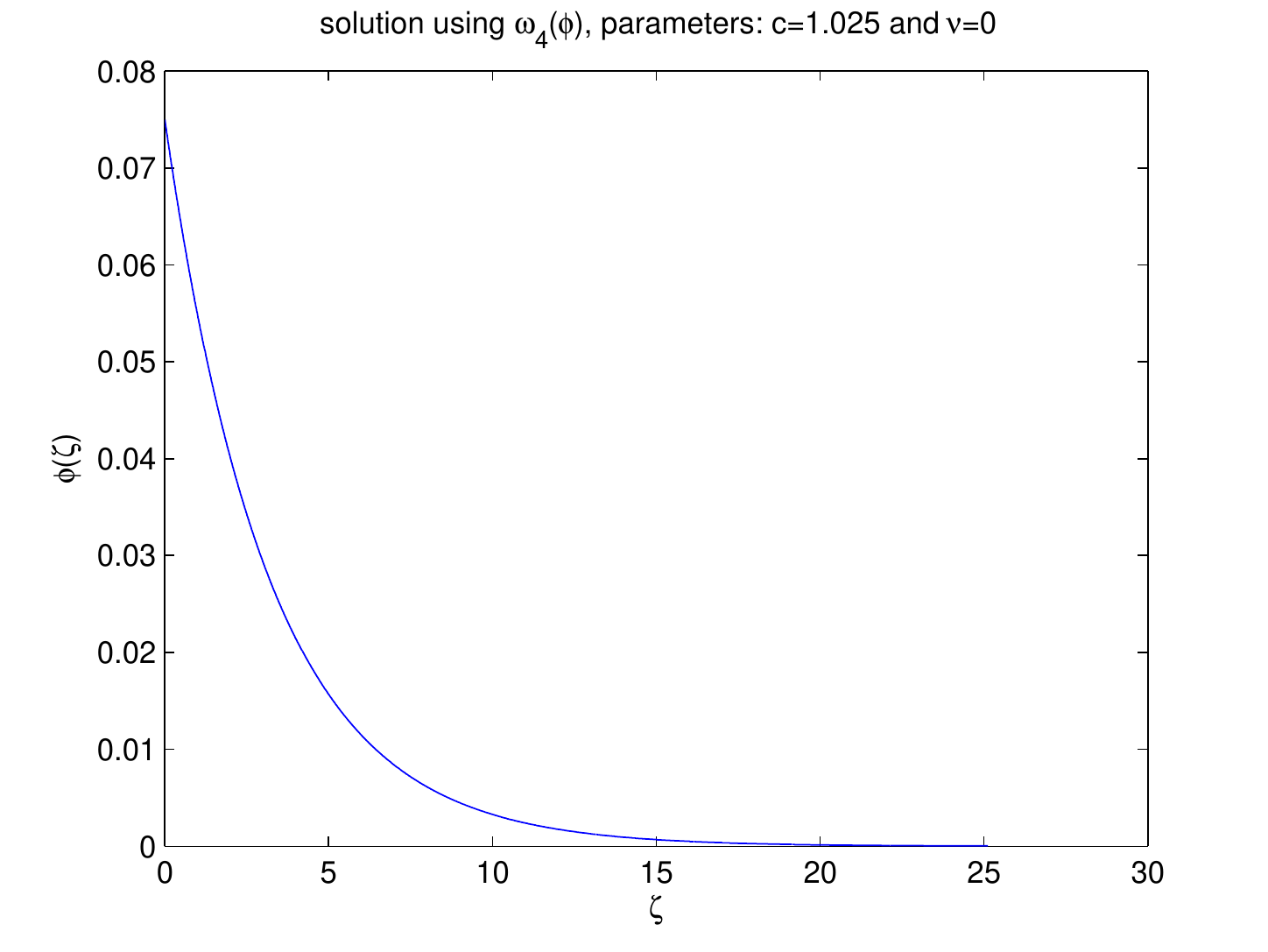}
					\end{center}
				\caption{Numerical solutions to \eqref{BBM} (Algebraic Method) predicting the traveling waves
solutions for $c>1$.}\label{fig4}
\end{figure}
\begin{figure}[!ht]
					\begin{center}
\includegraphics[width=0.40\textwidth]{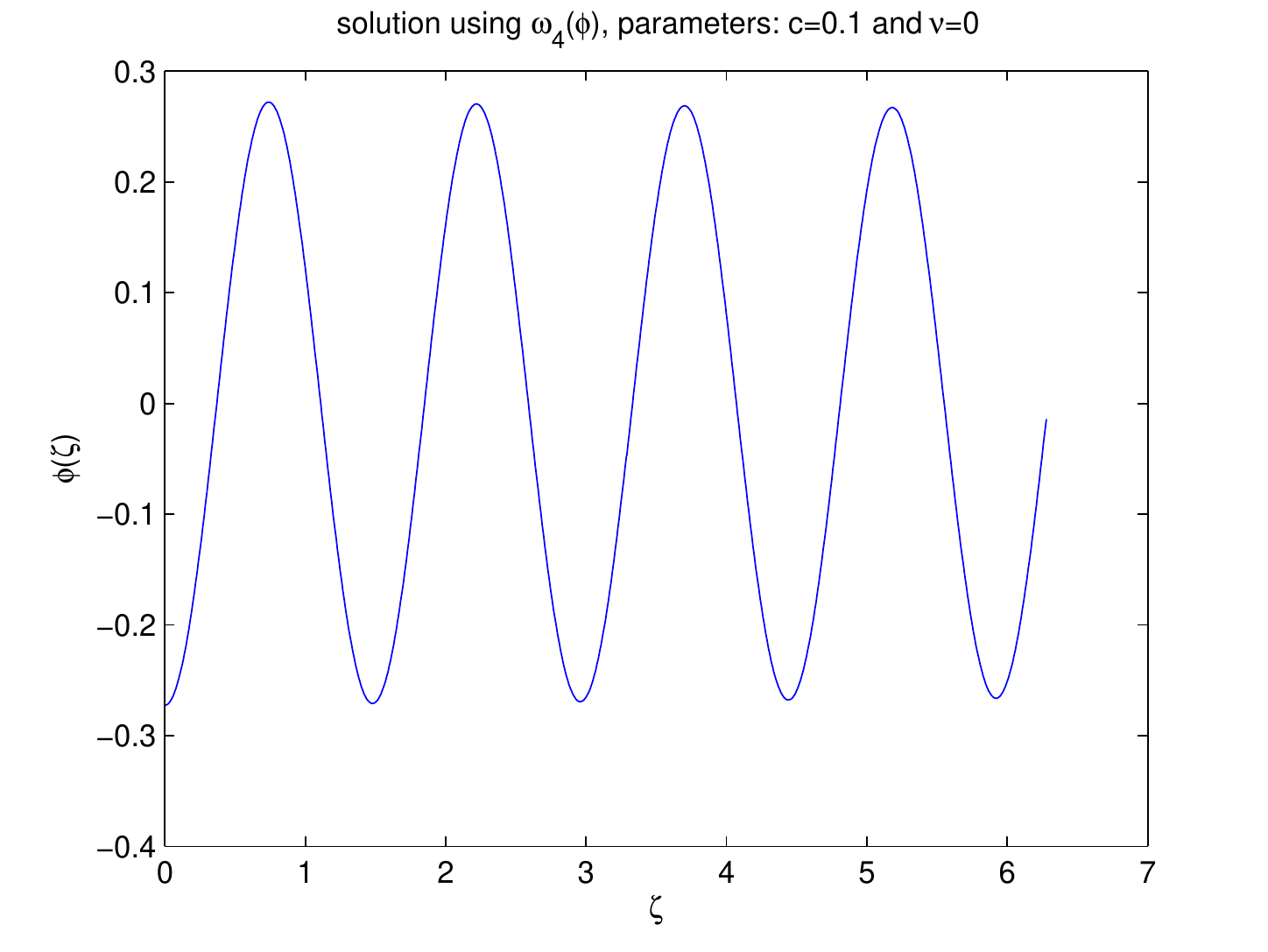}
					\end{center}
				\caption{Numerical solutions to \eqref{BBM} (Algebraic Method) predicting the
periodic solutions for $c<1$.}\label{fig5}
\end{figure}
\begin{figure}[!ht]
					\begin{center}
\includegraphics[width=0.40\textwidth]{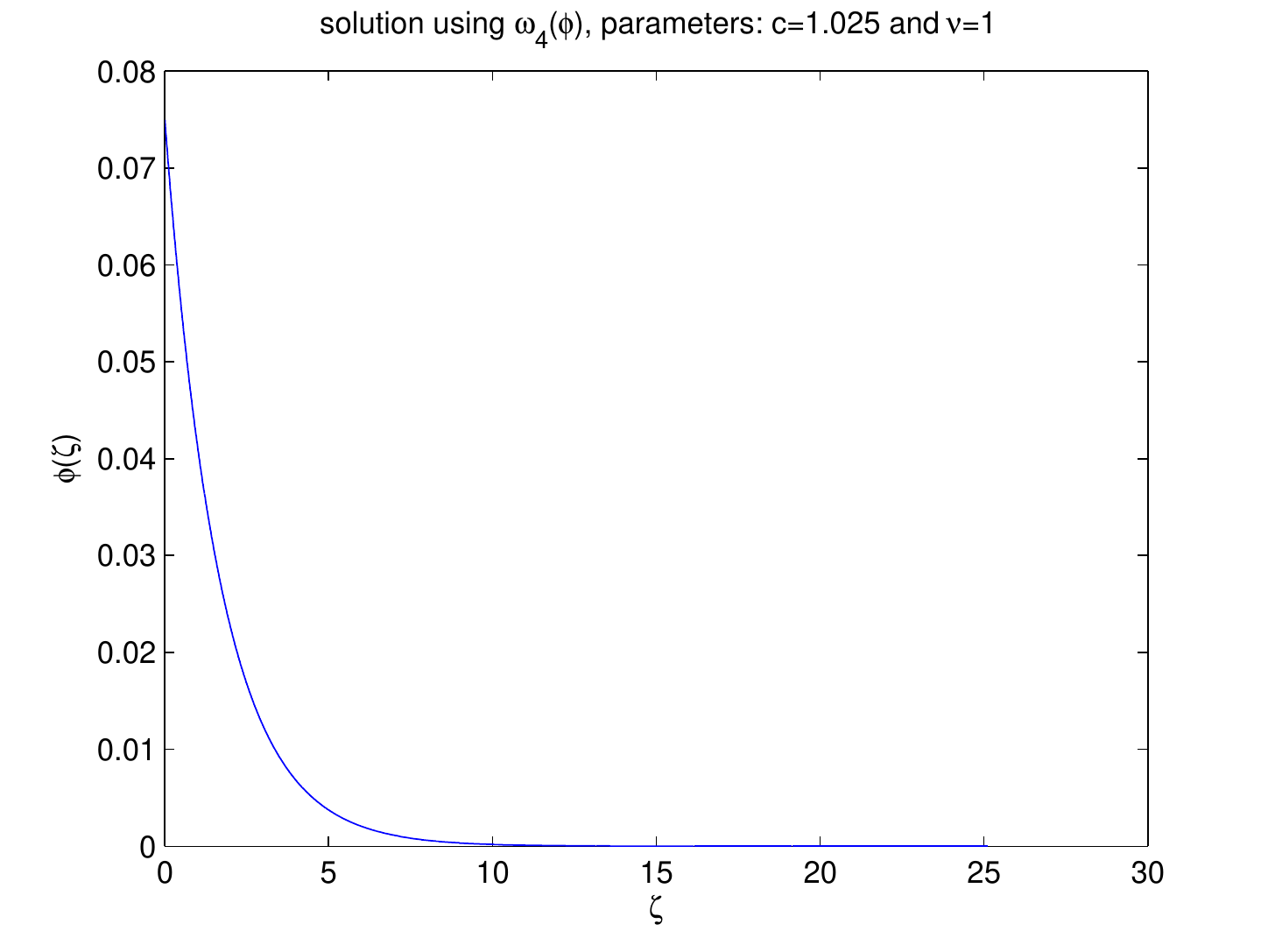}
					\end{center}
				\caption{Numerical solutions to \eqref{BBMv} (Algebraic Method) predicting the traveling waves solutions for $c>1$.}\label{fig6}
\end{figure}
\begin{figure}[!ht]
					\begin{center}
\includegraphics[width=0.40\textwidth]{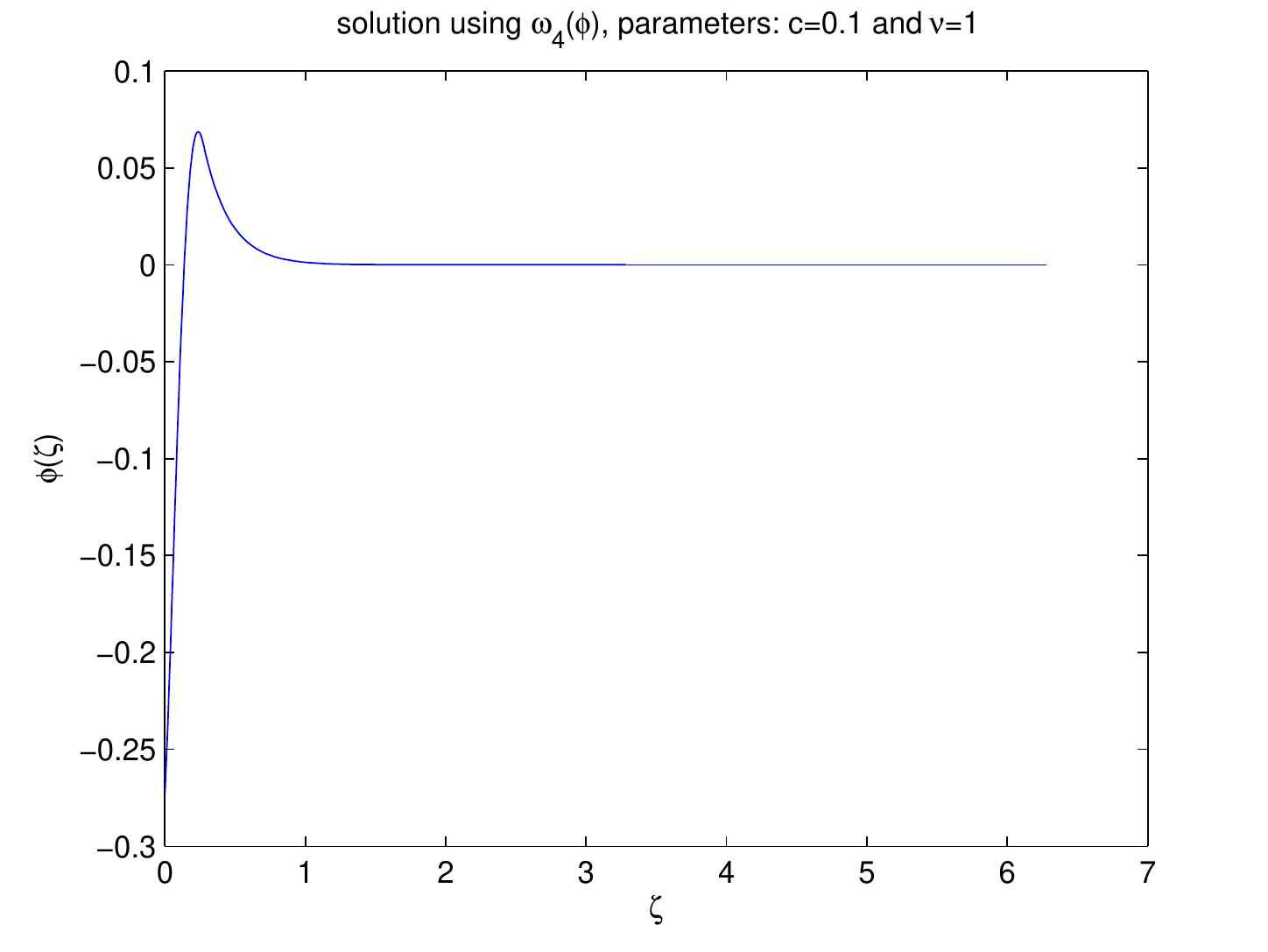}
					\end{center}
				\caption{Numerical solutions to \eqref{BBMv} (Algebraic Method) predicting the periodic solutions for $c<1$.}\label{fig7}
\end{figure}

In the absence of viscosity $\nu=0$, we obtained both traveling wave solution for $c>1$ in Fig.\ref{fig4}, and the periodic solution for $c>1$ in Fig.\ref{fig5}. In the presence of viscosity, $\nu=1$ we expect some damping effect as time progresses. This is depicted in Figs.\ref{fig6},\ref{fig7}. 

From a qualitative point of view Fig.\ref{fig4} is to be compared to Fig.\ref{fig1} and Fig.\ref{fig5} to Fig.\ref{fig2}.

\section{Stability of the viscous waves}

In this section we consider the two-mode dynamical system of \eqref{6}
\begin{align}\label{17}
\dot{\phi}&=\psi\\
\dot{\psi}&=\frac{\alpha}{2}\psi-\frac{\gamma}{2}\phi-\frac{\beta}{2} \phi^2\notag
\end{align}
with equilibrium points in the phase plane $(\phi,\psi)$ at $(0,0)$ and\\ $(-\frac{\gamma}{\beta},0)=(2(c-1),0)$. All the equilibrium points lie only in the plane $\psi=0$, in the $(\phi,\psi,c)$ space. The bifurcation curves are given by $\phi\big(\phi-2(c-1)\big)=0$, which are two straight lines intersecting at $c=1$. There is a bifurcation point at $c=1$, since the number of equilibrium points changes from two ($0<c<1$), to one ($c=1$)  and back to two ($c>1$), as $c$ increases.

Near the origin $(\phi_0,\psi_0)=(0,0)$
\begin{eqnarray}
\left[\begin{array}{cc}
\dot{\phi}\\
\dot{\psi}\\
\end{array}\right]\approx\left[\begin{array}{cc}
0& 1\\
-\frac{\gamma}{2} & \frac{\alpha}{2}\\
\end{array}\right]\left[\begin{array}{cc}
\phi\\
\psi\\
\end{array}\right]=\left[\begin{array}{cc}
0& 1\\
-\frac{1-c}{c} & \frac{\nu}{c}\\
\end{array}\right]\left[\begin{array}{cc}
\phi\\
\psi\\
\end{array}\right].
\end{eqnarray}
Following standard methods of phase--plane analysis the characteristic polynomial of the Jacobian matrix of \eqref{17} evaluated at the fixed point $(\phi_0,\psi_0)$ is
\begin{equation} \label{charac1}
g_0(\lambda)=\lambda^2-p_0\lambda +q_0=0,
\end{equation}
where $p_0=\frac{\nu}{c}$, and $q_0=\frac{1-c}{c}$.

Since $\nu>0$, $c>0$, then $p_0>0$, hence the origin is \underline{unstable}. Also, putting $\Delta=p_0^2-4q_0=\frac{\nu^2-4c(1-c)}{c^2}$,
we have the following cases:
\begin{itemize}
\item[(i)] $0<c<1$, gives $q_0>0$. If $\nu>2\sqrt{c(1-c)}>0$, the origin is unstable \emph{node} see Fig. \ref{fig8}; if $0<\nu<2\sqrt{c(1-c)}$, the origin is unstable \emph{spiral}, see Fig. \ref{fig9}.
\item[(ii)] $c>1$, gives $q_0<0\Rightarrow \Delta >0$, hence the origin is a \textit{saddle} point, see Figs. \ref{fig10},\ref{fig11}.
\end{itemize}

\begin{figure}[!ht]
					\begin{center}
\includegraphics[width=0.5\textwidth]{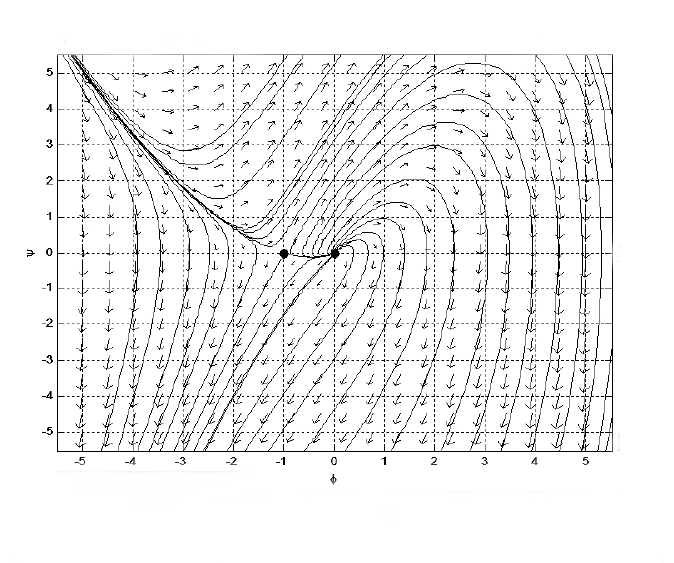}
					\end{center}
				\caption{Phase-plane for $c=0.5, \nu=1, (0.0)$ node, $(-1,0)$ saddle}\label{fig8}
\end{figure}
\begin{figure}[!ht]
					\begin{center}
\includegraphics[width=0.5\textwidth]{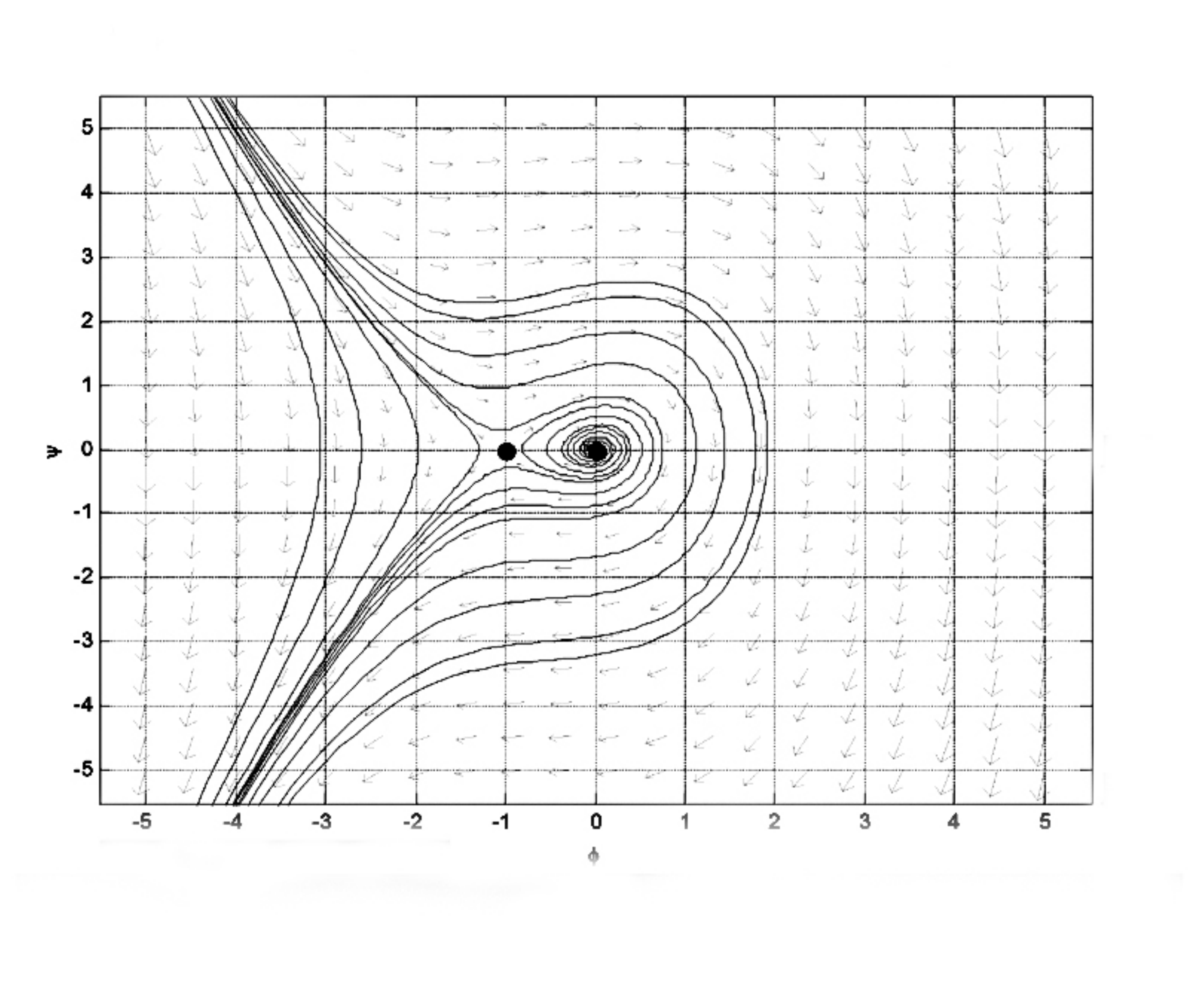}
					\end{center}
				\caption{Phase-plane for $c=0.5, \nu=0.1, (0.0)$ spiral, $(-1,0)$ saddle}\label{fig9}
\end{figure}

Near the secondary fixed point $(\phi_1,\psi_1)=(2(c-1),0)$
\begin{eqnarray}
\left[\begin{array}{cc}
\dot{\phi}\\
\dot{\psi}\\
\end{array}\right]\approx\left[\begin{array}{cc}
0& 1\\
-\frac{\gamma}{2}-2\beta(c-1) & \frac{\alpha}{2}\\
\end{array}\right]\left[\begin{array}{cc}
\phi\\
\psi\\
\end{array}\right]=\left[\begin{array}{cc}
0& 1\\
\frac{1-c}{c} & \frac{\nu}{c}\\
\end{array}\right]\left[\begin{array}{cc}
\phi\\
\psi\\
\end{array}\right].
\end{eqnarray}
The characteristic polynomial of the Jacobian matrix of \eqref{17} evaluated at the fixed point $(\phi_1,\psi_1)$ is
\begin{equation} \label{charac2}
g_1(\lambda)=\lambda^2-p_1\lambda +q_1=0,
\end{equation}
where $p_1=p_0=\frac{\nu}{c}$, and $q_1=-q_0=-\frac{1-c}{c}$.

Since $p_1=p_0$, the second fixed point is also \underline{unstable}, and moreover we have the cases:

\begin{itemize}
\item[(i)] $0<c<1$, gives $q_1<0\Rightarrow \Delta>0$, hence the fixed point is a \textit{saddle} point, see Figs. \ref{fig8},\ref{fig9}.
\item[(ii)] $c>1$, gives $q_1>0$. If $0<\nu<2\sqrt{c(c-1)}$, the fixed point is unstable \emph{spiral}, Fig. \ref{fig11} or an unstable \emph{node} for $\nu>2\sqrt{c(c-1)}$, Fig. \ref{fig10}.
\end{itemize}
All the other remaining cases, i.e, both fixed points collide, are degenerate, since if $c=1$, $g(\lambda)=\lambda^2-\nu \lambda$.

\begin{figure}[!ht]
					\begin{center}
\includegraphics[width=0.5\textwidth]{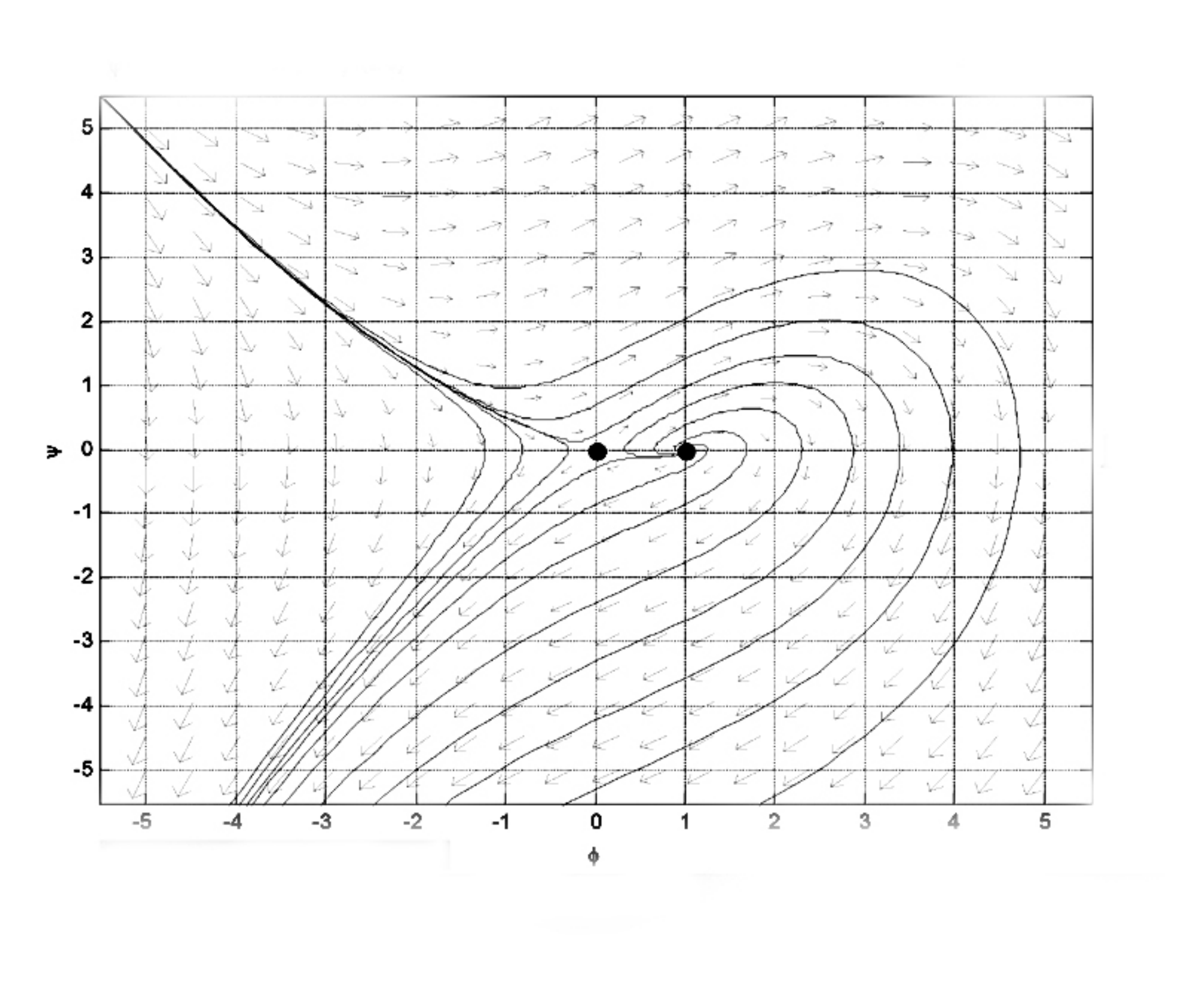}
					\end{center}
				\caption{Phase-plane for $c=1.5, \nu=1, (0.0)$ saddle, $(1,0)$ node}\label{fig10}
\end{figure}
\begin{figure}[!ht]
					\begin{center}
\includegraphics[width=0.5\textwidth]{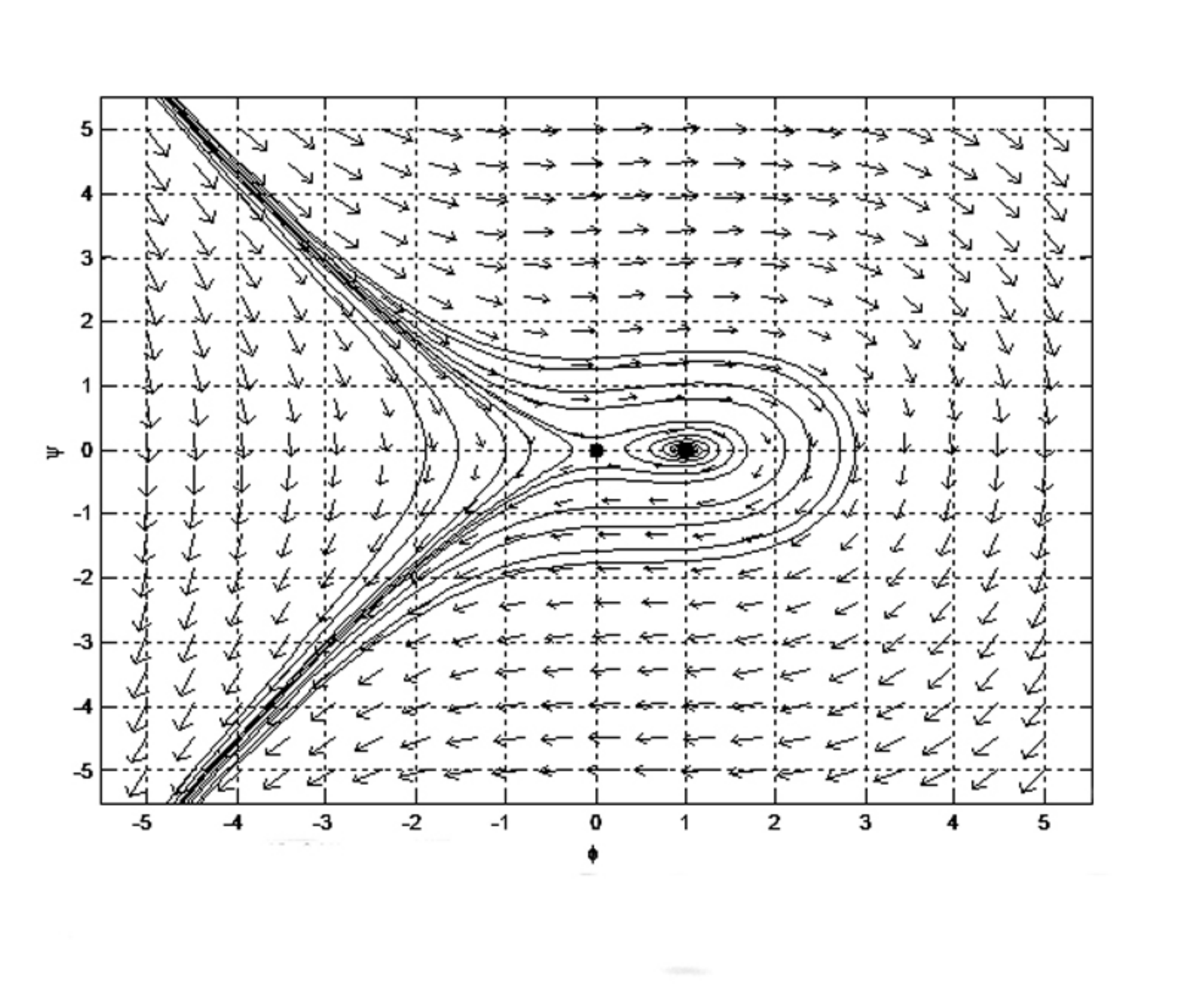}
					\end{center}
				\caption{Phase-plane for $c=1.5, \nu=0.1, (0.0)$ saddle, $(1,0)$ spiral} \label{fig11}
\end{figure}

Note that when there is no viscosity, the system has Hamiltonian $H(\phi,\psi)=\frac{1}{2}\psi^2+\frac{\gamma}{4}\phi^2+\frac{\beta}{6}\phi^3$. Therefore, along any phase path $H(\phi,\psi)=constant$. In this case, $p_0=p_1=0$, and hence
\begin{itemize}
\item[(i)] $(\phi_0,\psi_0)$ is a \emph{saddle} and $(\phi_1,\psi_1)$ is a \emph{center} when $q_1>0$, and
\item[(ii)] $(\phi_0,\psi_0)$ is a \emph{center} and $(\phi_1,\psi_1)$ is a \emph{saddle} when $q_1<0$.
\end{itemize}
Therefore, in this case the unstable spirals from Figs. \ref{fig9}, \ref{fig11} (which correspond to the case with small viscosity $\nu=0.1$) become centers when there is no viscosity $\nu=0$.

This is an example of a transcritical bifurcation where, at the intersection of the two bifurcation curves $\phi=0$ and $\phi=2(c-1)$, the equilibrium changes from one curve to the other at the bifurcation point. As $c$ increases through one,  the saddle point collides with the unstable node, and then remains there whilst the unstable node or spiral moves away from $(\psi_1,\phi_1)$.

\section{Summary and conclusions}
In this paper a basic theory for the BBM equation \eqref{BBM}, and its extension \eqref{BBMv} to include the dissipation term was shown. When the viscous terms are not present, \eqref{BBM} has traveling wave solutions that depend critically on the traveling wave velocity. When the velocity is sub unitary, it has solitary wave like solutions. If the velocity is super unitary, the solutions become periodic and unbounded. At the interface between the two cases, when $c=1$, \eqref{BBM} has periodic solutions in terms of the elliptic functions. Also, and ad-hoc theory based on the dissipative term was presented, in which we have found a set of solutions in terms of an implicit function that was solved numerically in MATLAB. There is a good agreement in the qualitative behavior of the solutions obtained numerically using the algebraic method and the closed form solutions.  Based on dynamical systems theory we have proved that the solutions of \eqref{BBMv} experience a transcritical bifurcation when $c=1$, where at the intersection of the two bifurcation curves, stable equilibrium changes from one curve to the other at the bifurcation point. As the velocity changes, the saddle point collides with the node at the origin, and then remains there, while the stable node moves away from the origin.



\subsection*{Acknowledgment}
The authors would like to acknowledge extremely insightful comments by David Ross on the theory of shallow water waves.
\end{document}